%% file: main.tex
\documentclass[thm-restate]{llncs}

\usepackage{geometry}
\usepackage[english]{babel}
\usepackage[utf8]{inputenc}
\usepackage[T1]{fontenc}
\usepackage{microtype,amssymb}

\usepackage{amsmath}
\usepackage{mathtools}
\usepackage{graphicx}
\usepackage[cmyk,dvipsnames]{xcolor}
\usepackage[hidelinks]{hyperref}
\usepackage[capitalise,noabbrev]{cleveref}
\usepackage{environ}
\usepackage{tikz}
\usetikzlibrary{calc}

\newtheorem{observation}{Observation}

\newcommand{\norm}[1]{\left|#1\right|}          
\newcommand{\nT}[0]{\overline{T}}               

\usepackage{color}

\newcommand{\say}[1]{``#1''} 

\DeclareMathOperator{\occ}{count}

\newcommand{\dd}{\mbox{---}}
\newcommand{\repeattheorem}[1]{%
    \begingroup
    \renewcommand{\thetheorem}{\ref{#1}}%
    \expandafter\expandafter\expandafter\theorem
    \csname reptheorem@#1\endcsname
    \endtheorem
    \endgroup
}

\NewEnviron{reptheorem}[1]{%
    \global\expandafter\xdef\csname reptheorem@#1\endcsname{%
        \unexpanded\expandafter{\BODY}%
    }%
    \expandafter\theorem\BODY\unskip\label{#1}\endtheorem
}

\newcommand{\repeatlemma}[1]{%
    \begingroup
    \renewcommand{\thelemma}{\ref{#1}}%
    \expandafter\expandafter\expandafter\lemma
    \csname replemma@#1\endcsname
    \endlemma
    \endgroup
}

\NewEnviron{replemma}[1]{%
    \global\expandafter\xdef\csname replemma@#1\endcsname{%
        \unexpanded\expandafter{\BODY}%
    }%
    \expandafter\lemma\BODY\unskip\label{#1}\endlemma
}

\newcommand{\repeatobservation}[1]{%
    \begingroup
    \renewcommand{\theobservation}{\ref{#1}}%
    \expandafter\expandafter\expandafter\observation
    \csname repobservation@#1\endcsname
    \endobservation
    \endgroup
}

\NewEnviron{repobservation}[1]{%
    \global\expandafter\xdef\csname repobservation@#1\endcsname{%
        \unexpanded\expandafter{\BODY}%
    }%
    \expandafter\observation\BODY\unskip\label{#1}\endobservation
}

\renewcommand{\P}{\mathbb{P}}
\newcommand{\E}{\mathbb{E}}

\title{Reconstructing Graphs from Connected Triples\thanks{LC is supported by the Institute for Basic Science (IBS-R029-C1) and CG by 
Marie-Skłodowska Curie grant GRAPHCOSY (number 101063180). MvK and JV are supported by the Netherlands Organisation for Scientific Research (NWO) under
project no. 612.001.651.}}

\author{Paul Bastide\inst{1}
\and Linda Cook\inst{2} \and Jeff Erickson\inst{3} \and Carla Groenland\inst{4} \and \\ Marc van Kreveld\inst{4} \and Isja Mannens\inst{4} \and Jordi L. Vermeulen\inst{4}}
\institute{
LaBRI - Bordeaux University\\
\email{paul.bastide@ens-rennes.fr}\and
Institute for Basic Science, Discrete Math Group, Republic of Korea\\
\email{lindacook@ibs.re.kr}\and
    University of Illinois at Urbana-Champaign\\
    \email{jeffe@illinois.edu}\and
    Utrecht University\\
    \email{\{c.e.groenland,m.j.vankreveld,i.m.e.mannens\}@uu.nl}
}

\date{March 12, 2022}

\begin{document}
\maketitle
\begin{abstract}
We introduce a new model of indeterminacy in graphs: instead of specifying all the edges of the graph, the input contains all triples of vertices that form a connected subgraph.
In general, different (labelled) graphs may have the same set of connected triples, making unique reconstruction of the original graph from the triples impossible.
We identify some families of graphs (including triangle-free graphs) for which all graphs have a different set of connected triples.
We also give algorithms that reconstruct a graph from a set of triples, and for testing if this reconstruction is unique.
Finally, we study a possible extension of the model in which the subsets of size $k$ that induce a connected graph are given for larger (fixed) values of $k$.
\end{abstract}

\section{Introduction}
Imagine that we get information about a graph, but not its complete structure by a list of edges.
Does this information uniquely determine the graph?
In this paper we explore the case where the input consists of all triples of vertices whose induced subgraph is connected.
In other words, we know for each given triple of vertices that two or three of the possible edges are present, but we do not know which ones.
We may be able to deduce the graph fully from all given triples.

As a simple example, assume we are given the (unordered, labelled) triples \(abc\), \(bcd\), and \(cde\).
Then the only (connected) graph that matches this specification by triples is the path \(a\)---\(b\)---\(c\)---\(d\)---\(e\).
On the other hand, if we are given all possible triples on a set of four vertices \(a,b,c,d\) except for \(abc\), then there are several graphs possible. 
We must have the edges \(ad\), \(bd\), and \(cd\), and zero or one of the edges \(ab\), \(bc\), and~\(ca\). See Figure \ref{fig:label-ambiguity} for another example.

\begin{figure}[t]
    \centering
    \includegraphics{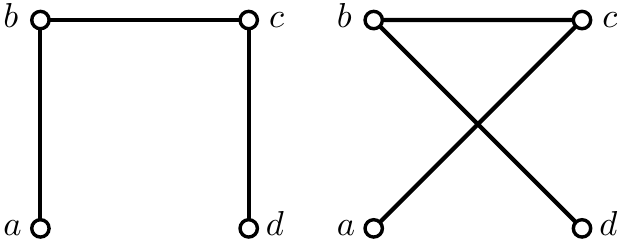}
    \caption{Two different labelled trees that give the same set of connected triples.}
    \label{fig:label-ambiguity}
\end{figure}

This model of indeterminacy of a graph does not use probability and is perhaps the simplest combinatorial model of partial information.
Normally a graph is determined by pairs of vertices which are the edges; now we are given triples of vertices with indeterminacy on the edges between them.
As such, we believe this model is interesting to study.

As illustrated in our previous example, there are cases where reconstruction of the graph from the set \(T\) of triples is unique and there are cases where multiple (labelled) graphs may have the same set of triples.
There are also cases where \(T\) is not consistent with any graph, such as \(T=\{abc,cde\}\).
Can we characterize these cases, and what can we say if we have additional information, for example, when we know that we are reconstructing a tree or a triangle-free graph?

\subsection{Our results}
After preliminaries in \cref{sec:prelim}, we provide two relatively straightforward, general algorithms for reconstruction in \cref{sec:algorithm}.
One runs in \(O(n^3)\) time when the triples use \(n\) vertex labels, and the other runs in \(O(n\cdot |T|)\) time when there are \(|T|\) triples in the input.
These algorithms return a graph that is consistent with the given triples, if one exists, and decide on uniqueness.

Then, we give an \(O(|T|)\) time algorithm to reconstruct trees on at least five vertices, provided that the unknown graph is known to be a tree, in \cref{sec:trees}.
In fact, all triangle-free graphs can be reconstructed, provided we know that the unknown graph is triangle-free. We give an algorithm running in expected $O(|T|)$ time for this in \cref{sec:unique-graphs}.
Moreover, we show that 2-connected outerplanar graphs and triangulated planar graphs can be uniquely reconstructed.

In \cref{sec:extensions} we study a natural extension of the model where we are given the connected $k$-sets of a graph for some fixed $k\geq 4$, rather than the connected triples. We show the largest value of $k$ such that each $n$-vertex tree has a different collection of connected $k$-sets is $\lceil n/2\rceil$. A similar threshold is shown for the random graph. Finally, we show that graphs without cycles of length $3,\dots,k$ on at least $2k-1$ vertices have different collections of connected $k$-sets.

\subsection{Related work}\label{sec:graph-related-work}
The problem of graph reconstruction arises naturally in many cases where some unknown graph is observed indirectly.
For instance, we may have some (noisy) measurement of the graph structure, or only have access to an oracle that answers specific types of queries.
Much previous research has been done for specific cases, such as reconstructing metric graphs from a density function~\cite{dey2018morse}, road networks from a set of trajectories~\cite{ahmed2012trajectories}, graphs using a shortest path or distance oracle~\cite{kannan2018reconstruction}, labelled graphs from all \(r\)-neighbourhoods~\cite{mossel2017shotgun}, or reconstructing phylogenetic trees~\cite{brandes2009phylogenetic}. 
A lot of research has been devoted to the \emph{graph reconstruction conjecture} \cite{Kelly42,Ulam60}, which states that it is possible to reconstruct any graph on at least three vertices (up to isomorphism) from the multiset of all (unlabeled) subgraphs obtained through the removal of one vertex. This conjecture is open even for planar graphs and triangle-free graphs, but has been proved for outerplanar graphs~\cite{Giles74} and maximal planar graphs~\cite{Lauri81}. We refer the reader to one of the many surveys (e.g.~\cite{bondy1977reconstruction,harary1974conjecture,tutte,LS16}) for further background.
Related to our study of the random graph in Section \ref{sec:extensions} is a result from Cameron and Martins \cite{cameron_martins_1993} from 1993, which implies that for each graph $H$, with high probability the random graph $G\sim G(n,\frac12)$ can be reconstructed from the set of (labelled) subsets that induce a copy of $H$ (up to complementation if $H$ is self-complementary).

Many types of uncertainty in graphs have been studied.
Fuzzy graphs~\cite{rosenfeld1975fuzzy} are a generalisation of fuzzy sets to relations between elements of such sets.
In a fuzzy set, membership of an element is not binary, but a value between zero and one.
Fuzzy graphs extend this notion to the edges, which now also have a degree of membership in the set of edges.
Uncertain graphs are similar to fuzzy graphs in that each edge has a number between zero and one associated with it, although here this number is a probability of the edge existing.
Much work has been done on investigating how the usual graph-theoretic concepts can be generalised or extended to fuzzy and uncertain graphs~\cite{kassiano2017mining,mordeson1994operations}.

\section{Preliminaries}\label{sec:prelim}
All graphs in this paper are assumed to be connected, finite, and simple. Let \(G\) be an unknown graph with \(n\) vertices and let \(T\) be the set of all triples of vertices that induce a connected subgraph in \(G\). Since the graph is connected, we can recover the vertex set $V$ of $G$ easily from $T$.
We will use \(\nT\) to denote the complement of this set \(T\), i.e.
\(\nT\) is the set of all triples of vertices for which the induced subgraph is not connected.
Note that \(\norm{T \cup \nT} = \binom{n}{3} \in \Theta(n^3)\).

Observe that both the presence and absence of a triple gives important information:
in the former case, at most one of the three possible edges is absent, whereas in the latter case, at most one of these edges is present.

\begin{figure}[t]
    \centering
    \includegraphics{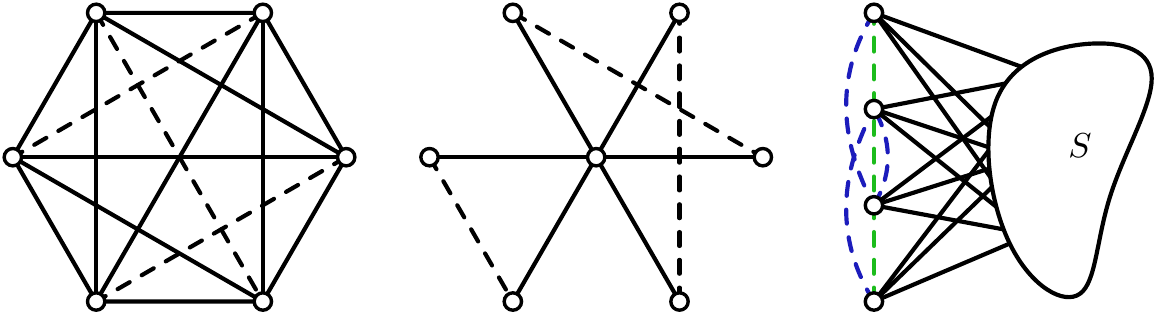}
    \caption{Three classes of ambiguous triples: a complete graph minus any independent set of edges, a star graph plus any (partial) matching of the leaves, and a path of length four in which all vertices are fully adjacent to some set \(S\).
    In this last case, we cannot tell the difference between the blue and green path.}
    \label{fig:ambiguous-classes}
\end{figure}

It is possible that graphs that are not the same (as labelled graph) or even not isomorphic yield the same set of triples, for example, a path on three vertices and a triangle. We also give examples of larger graphs that cannot be distinguished from their set of connected triples in Figure~\ref{fig:ambiguous-classes}.

We will make use of (LSD) radix string sorting (as described in e.g. \cite{cormen2022introduction}) to sort a collection of $t$ sets of cardinality $k$ in time $\mathcal{O}(tk)$ in several of the algorithms presented in this paper.

\section{Algorithm for finding consistent graphs from triples}
\label{sec:algorithm}

Given a set of triples \(T\), we can find a graph \(G\) consistent with those triples by solving a 2-SAT formula.
The main observation here is that the presence of a triple \(abc\) means that at least two of the edges \(ab\), \(ac\) and \(bc\) must exist, whereas the absence of a triple means at most one of the edges can exist.
We can then construct a 2-SAT formula where each variable corresponds to an edge of the graph, and truth represents presence of that edge.
For each triple \(abc \in T\), we add clauses \((ab \vee ac)\), \((ab \vee bc)\) and \((ac \vee bc)\) to the formula.
For each triple \(abc \in \nT\), we add clauses \((\neg ab \vee \neg ac)\), \((\neg ab \vee \neg bc)\) and \((\neg ac \vee \neg bc)\).
A graph consistent with the set of triples can then be found by solving the resulting 2-SAT formula and taking our set of edges to be the set of true variables in the satisfying assignment.
If the formula cannot be satisfied, no graph consistent with \(T\) exists.

We can solve the 2-SAT formula in linear time with respect to the length of the formula~\cite{aspvall1979sat,even1976sat}.
We add a constant number of clauses for each element of \(T\) and \(\nT\), so our formula has length \(O(\norm{T \cup \nT})\).
As \(\norm{T \cup \nT} = \binom{n}{3}\), this gives us an \(O(n^3)\) time algorithm to reconstruct a graph with \(n\) vertices.
However, we prefer an algorithm that depends on the size of \(T\), instead of also on the size of \(\nT\).
We can eliminate the dependency on the size of \(\nT\) by observing that some clauses can be excluded from the formula because the variables cannot be true.
\begin{lemma}\label{lem:algorithm}
    We can find a graph \(G\) consistent with \(T\) in \(O(n \cdot \norm{T})\) time, or output that no consistent graph exists.
\end{lemma}
\begin{proof}
    The basic observation that allows us to exclude certain clauses from the formula is that if there is no connected triple containing two vertices \(a\) and \(b\), the variable \(ab\) will always be false. 
    Consequently, if we have a triple \(abc \in \nT\) for which at most one of the pairs \(ab\), \(ac\) and \(bc\) appear in some connected triple, we do not need to include its clauses in the formula, as at least two of the variables will be false, making these clauses necessarily satisfied.

    We can construct the formula that excludes these unnecessary clauses in \(O(n \cdot \norm{T})\) time as follows.
    We build a matrix \(M(i, j)\), where $i, j \in V(G)$ with each entry containing a list of all vertices with which \(i\) and \(j\) appear in a connected triple, i.e.
    \(M(i, j) = \{x~|~ijx \in T\}\).
    This matrix can be constructed in \(O(n^2 + \norm{T})\) time, by first setting every entry to $\emptyset$ (this takes $n^2$ time) and then running through $T$, adding every triple $abc$ to the entries $M(a,b)$, $M(b,c)$ and $M(a,c)$. 
    We also sort each list in linear time using e.g.
    radix sort.
    As the total length of all lists is \(O(\norm{T})\), this takes \(O(n^2 + \norm{T})\) time in total. 

    Using this matrix, we can decide which clauses corresponding to triples from \(\nT\) to include as follows.
    For all pairs of vertices \((a,b)\) that appear in some connected triple (i.e.
    \(M(a, b) \neq \emptyset\)), we find all \(x\) such that \(abx \in \nT\).
    As \(M(a, b)\) is sorted, we can find all \(x\) in \(O(n)\) time by simply recording the missing elements of the list $M(a,b)$. 
    We then check if \(M(a, x)\) and \(M(b, x)\) are empty.
    If either one is not, we include the clause associated with \(abx \in \nT\) in our formula.
    Otherwise, we can safely ignore this clause, as it is necessarily satisfied by the variables for \(ax\) and \(bx\) being false.
    
    Our algorithm takes \(O(n)\)-time for each non-empty element of \(M(i, j)\), of which there are \(O(\norm{T})\), plus \(O(n^2)\) time to traverse the matrix.
    The total time to construct the formula is \(O(n^2 + n \cdot \norm{T})\).
    As \(\norm{T} \in \Omega(n)\) for connected graphs,
    this simplifies to \(O(n \cdot \norm{T})\) time.
    The resulting formula also has \(O(n \cdot \norm{T})\) length, and can be solved in time linear in that length.
\qed \end{proof}
Observe that this is only an improvement on the naive \(O(n^3)\) approach if \(\norm{T} \in o(n^2)\).
We also note that we can test the uniqueness of the reconstruction in the same time using Feder's approach for enumerating 2-SAT solutions~\cite{feder1994sat}.

\section{Unique reconstruction of trees}\label{sec:trees}
In this section, we prove the following result.
\begin{theorem}\label{thm:tree-reconstruction}
    Let \(T\) be a set of triples, and let it be known that the underlying graph \(G = (V, E)\) is a tree.
    If \(n \geq 5\), then \(G\) can be uniquely reconstructed in \(O(|T|)\) time.
\end{theorem}
Let us briefly examine trees with three or four vertices.
A tree with three vertices is always a path and it will always have one triple with all three vertices.
We do not know which of the three edges is absent.
A tree with four vertices is either a path or a star.
The path has two triples and the star has three triples.
For the star, the centre is the one vertex that appears in all three triples, and hence the reconstruction is unique.
For the path, we will know that the graph is a path, but we will not know in what order the middle two vertices appear (see Figure \ref{fig:label-ambiguity}).

\begin{figure}[t]
    \centering
    \includegraphics{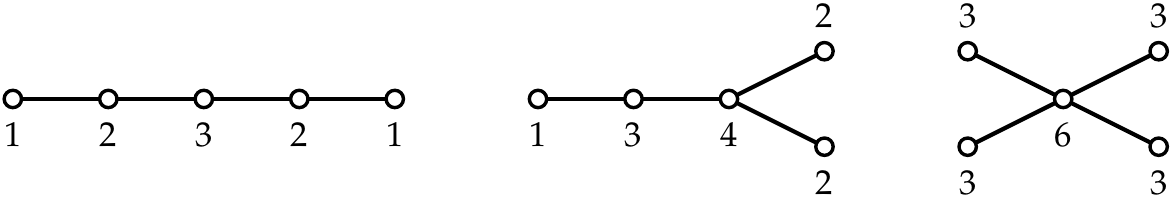}
    \caption{All trees on five vertices, and the number of triples each vertex occurs in.}
    \label{fig:5-vertex-trees}
\end{figure}

Next we consider trees with at least five vertices.
We first show that we can recognise all leaves and their neighbours from the triples.
In the following, we say that a vertex \(v\) \emph{dominates} a vertex \(u\) if \(v\) appears in all the triples that \(u\) appears in.
If \(u\) is a leaf, then it is dominated by its unique neighbour \(v\). It is possible that \(u\) is also dominated by a neighbour of \(v\), but in this case \(uvw\) will be the only triple containing \(u\), and \(v\) will be dominated by \(w\). This can be used to recognise the leaves.
We can use this to prove that any tree can be reconstructed from its triples, provided that we know that the result must be a tree and \(|V|\geq 5\), since we can iteratively recognise and remove vertices of degree 1, while recording where to `glue them back at the end' until at most four vertices remain. We can complete the reconstruction via some closer examination of the connected triples in the original tree that contain the remaining vertices. 

In order to derive an optimal, \(O(|T|)\) time reconstruction algorithm, we will use a further characterisation of vertices of a tree using the triples.
The main idea is that we can recognise not only leaves, but also other vertices where we can reduce the tree.
If a vertex \(v\) has degree \(2\) in a tree, then there are two nodes \(w,w'\) such that every triple with \(v\) also contains \(w\) or \(w'\) (or both).
The converse is not true for two reasons: if \(v\) is a leaf, it also has the stated property, and if \(v\) has degree \(3\) where at least one neighbour is a leaf, then it has this property as well. This brings us to the following characterisation.
\begin{lemma}\label{lem:tree-alg}
    A vertex \(v\) of a tree \(G\) of at least five vertices with triple set \(T\) is:
    \begin{description}
        \item[(i)] a leaf if and only if \(v\) is dominated by some vertex \(w\) and does not dominate any vertex itself;
        \item[(ii)] if \(v\) is not a leaf, then \(v\) is (a) a node of degree \(2\), or (b) a node of degree \(3\) with at least one leaf neighbour, if and only if there are two nodes \(w_1,w_2\) such that all triples with \(v\) also contain \(w_1\) or \(w_2\).
    \end{description}
    Moreover, both characterisations can be checked in time \(O(|T_v|)\), when the set of triples \(T_v\) that include \(v\) is given for all $v\in V(G)$.
\end{lemma}

\begin{proof}
    A leaf \(v\) can necessarily only appear in triples with its adjacent vertex \(w\), as it is not adjacent to any other vertices by definition.
    A leaf is therefore always dominated by its neighbour \(w\).     Since \(|V|\geq 5\), \(v\) does not dominate any vertex.
    Conversely, suppose that \(v\) is dominated by some vertex \(w\) and does not dominate any other vertex. It is straightforward to check that \(v\) cannot be dominated if it has three neighbours or if it has two non-leaf neighbours. Since \(v\) does not dominate any vertices, it does not have a leaf neighbour. So \(v\) must be a leaf itself. This proves (i).
    
The second characterisation can be seen as follows.
If \(v\) has degree at least \(4\), then no \(w_1,w_2\) as in (ii) exist, which is easily verified by looking at the triples with \(v\) and its neighbours only.
Furthermore, if \(v\) has degree \(3\) and none of its neighbours are leaves, then again there are no such \(w_1,w_2\).
On the other hand, in case (a) the two neighbours can be taken as \(w_1,w_2\) and in case (b) \(w_1,w_2\) can be chosen to be two neighbours of $v$ so that the unique neighbour of $v$ that is not in $\{w_1, w_2\}$ is a leaf.

For testing (i), take any triple \(vab\in T_v\), and test both \(a\) and \(b\) separately if they are the sought \(w\). For testing (ii), take any triple \(vab\in T_v\).
If characterisation (ii) holds, then \(w_1\) must be \(a\) or \(b\).
We try both as follows: For $w \in \{a, b\}$ we remove all triples with \(w\) from \(T_v\).
Then $w = w_1$ if and only if all of the remaining triples of $T_v$ all contain some $w_2 \neq v$. We test this by looking at some remaining triple in $vcd \in T_v$ and testing whether either $c$ or $d$ is contained in every other remaining triple.
In total, we get four options to test for \(w_1\) and \(w_2\); each option is easily checked in \(O(|T_v|)\) time.
\qed \end{proof}

Note that more than half of the vertices of \(G\) satisfy one of the two characterisations of the lemma.
We next turn to the proof of Theorem \ref{thm:tree-reconstruction}. We will assume that there is a total order on the vertex set of $G$ and that the connected triples \(uvw\) are stored in an ordered tuple with \(u< v< w\). 
For each triple \(uvw\) in \(T\), we generate \(vwu\) and \(wuv\) as well.
We collect the triples with the same first vertex to generate \(T_v\) for all \(v\in V\). 
Then, for all \(v\in V\), we use \(T_v\) to test if \(v\) satisfies one of the conditions of Lemma~\ref{lem:tree-alg}.
The vertices of \(V\) partition into \(V'\), \(V''\), and \(V'''\), where \(V'\) contains the leaves, \(V''\) contains the vertices that are not leaves but satisfy the second condition of the lemma, and \(V'''=V\setminus (V'\cup V'')\).

For all leaves \(v\), we record the set of vertices that dominate $v$ in $O(|T_v|)$. This is a single vertex $w$ that is the neighbour of $v$ unless $T_v$ contains exactly one triple $vw'w$. In this case, $w$ dominates $w'$ or vice versa; the dominated one of the two is the neighbour of $v$. Hence, we can record all the leaves and their neighbours.
After this, we remove all triples containing a leaf from $T$. Let $G'$ be the graph obtained by removing all leaves and incident edges. Then the new triple set is the set of connected triples for this graph $G'$, and we can recover $G$ from $G'$.

For all vertices in \(V''\), note that they can no longer be vertices of degree \(3\) in $G'$, but they may have become leaves.
We test this and consider the subset \(W\subseteq V''\) of vertices that have not become leaves. 

The subgraph of $G'$ induced on $W$ consists of a disjoint union of paths. 
Let $v \in W$. Then $v$ has exactly two neighbours in $V(G')$ and they are the two vertices $w_1, w_2$ satisfying the second condition of the Lemma \ref{lem:tree-alg}. We can find these two vertices $w_1,w_2$ for each $v\in W$ in time $O(|T_v|)$.
In particular, we know all path components of $G'[W]$, as well as the unique vertices in $V(G')\setminus W$ that the endpoints of any such path are adjacent to.
Suppose that \(v_1\)---\(v_2\)---\(\dots\)---\(v_\ell\). is one of the path components of $G[W]$. 
Let \(x_1,x_2\in V(G')\setminus  W'\) such that \(x_1\) is the other neighbour of \(v_1\), and \(x_2\) is the other neighbour of \(v_k\) (in $G'$). We record the edges $x_1v_1$ and $x_2v_k$, as well as the edges and vertices in the path \(v_1,\dots,v_\ell\).
Then we replace each triple \(ux_1v_1\) by \(ux_1x_2\) and each triple \(v_kx_2u\) by \(x_1x_2u\). Afterwards, we discard all triples that contain any of \(v_1,\ldots,v_\ell\).
Let $G''$ be the graph obtained by deleting $v_1, \dots, v_\ell$ and adding the edge $\{x_1, x_2\}$. 
The resulting triple set is the triple set for $G''$, and we can recover $G$ from $G''$. We repeat this for all path components. We maintain throughout that the stored triple set corresponds to a tree, and that we can reconstruct the original tree $G$ from knowing this tree and the additional information that we record.

Finally, we also remove all leaves in \(V''\setminus W\) by discarding more triples, similar to the first leaf removal.
This process takes time linear in \(|T|\), and reduces the number of vertices occurring in \(T\) to half or less.
We repeat the process on the remaining tree until it has size five, at which point we can uniquely identify the structure of the tree by simply looking at the number of triples each vertex occurs in (see \cref{fig:5-vertex-trees}).
We may not remove all vertices of \(V'\) or \(V''\) if the remaining tree would be smaller than five vertices; in that case, we simply leave some of them in.
A standard recurrence shows that the total time used is \(O(|T|)\). This finishes the proof of \cref{thm:tree-reconstruction}.

We note that if the tree contains no leaves that are siblings, then we do not need to know that the graph is a tree for unique reconstruction. 

\section{Further reconstructible graph classes}\label{sec:unique-graphs}
In this section, we give larger classes of graphs for which the graphs that are determined by their set of connected triples. 
\begin{theorem}\label{thm:triangle-free}
There is an algorithm that reconstructs a graph \(G\) on $n\geq 5$ vertices that is known to be triangle-free from its set \(T\) of triples in deterministic $O(|T| \log(|T|))$ time or randomized $O(|T|)$ expected time.
\end{theorem}
\begin{proof}
    Let $T$ be the given list of connected triples. For every triple $abc$ we create three ordered copies $abc$, $acb$, $bca$. We then sort the list in lexicographical order in $O(|T|)$ time using radix sort. For every potential edge $ab$ that appears as the first two vertices of some triple we test whether it is an edge as follows.

    If we find two triples $abc$ and $abd$, we search for the triples $acd$ and $bcd$. If $ab \notin E(G)$, then we must have that $bc, ac, bd, ad \in E(G)$, and thus both $acd$ and $bcd$ are connected. Therefore, if either triple is not in the list, we know that $ab$ is an edge. Otherwise, we find that $abcd$ is a $C_4$. We then search for another vertex $e$ that appears in a triple with any of $a, b, c, d$ and reconstruct the labeling of the $C_4$ as follows. Suppose $G[\{a,b,c,d\}]$ induces a $C_4$ (with also this order, but this we don't know immediately). Assume w.l.o.g. that $a$ has the remaining vertex $e$ as a neighbour, then since $G$ is triangle-free, $e$ is not adjacent to $b$ and $d$. This means that $bde$ is known to be disconnected, whereas $abe,ade$ are connected. If $e$ is not a neighbour of $c$, then $ace$ is not connected. Hence, if any of $a,b,c,d$ has a private neighbour (one not adjacent to other vertices in the cycle), then we get the labeling of our $C_4$ (and find that $e$ is a private neighbour of $a$). If $e$ is adjacent to $c$ besides $a$, then we know the following two vertex sets also induce $C_4$'s: $abce,acde$. We know $e$ is adjacent to two out of $\{a,b,c\},\{a,d,c\}$ but not to $b$ and $d$. So we find $e$ is adjacent to $a$ and $c$ and also have found our labeling.

    If we only find one triple $abc$, we check for triples starting in $ac$ or $bc$. Since one of the three vertices involved must be adjacent to some other vertex $d$, one of these two potential edges must appear in at least two triples. We can then use the previous methods to reconstruct some subgraph containing $a, b$ and $c$. The result will tell us whether $ab$ is an edge or not.

    Note that we can search our list for a specific triple or a triple starting with a specific pair of vertices, in time $O(\log(|T|))$ using binary search. This means that the above checks can be done in $O(\log(|T|))$ time. By handling potential edges in the order in which they appear in the list, we only need to run through the list once, and thus we obtain a runtime of $O(|T|\log(|T|))$.

    Using a data structure like the one described by Fredman et al. \cite{Fredman1982}, we can query the required triples in $O(1)$ time and thus reconstruct the graph in $O(|T|)$ time. Creating this data structure is a randomized procedure and takes $O(|T|)$ expected time.
\qed \end{proof}

Note, we can also reconstruct some special cases of graphs which contain triangles.
For example, we can reconstruct a triangle if two of its vertices have neighbors that are non-adjacent to any other vertex of the triangle.
However, we cannot distinguish a triangle where exactly one vertex has a private neighbor from the star $K_{1,3}$. 
We are far from a characterisation of when we can reconstruct graphs that are not known to be triangle free. 
We will the following two special cases in appendix.

\begin{reptheorem}{twoconnected}
    Any graph on $n \geq 6$ vertices that is known to 2-connected and outerplanar can be reconstructed from its list of connected triples. 
\end{reptheorem}
Note, a graph is \emph{$k$-connected} if it has more than $k$ vertices and the graph cannot be disconnected by removing fewer than $k$ vertices.
Our approach is similar to the one for trees: we show that we can identify a vertex of degree two, and remove it from the graph by `merging' it with one of its neighbours.

Our approach is similar to the one for trees: we show that we can identify a vertex of degree two, and remove it from the graph by `merging' it with one of its neighbours.

A \emph{triangulated planar graph}, also called a \emph{maximal planar graph}, is a planar graph where every face (including the outer face) is a triangle.
\begin{reptheorem}{triangulated}
    Let \(T\) be a set of triples, and let it be known that the underlying graph \(G = (V, E)\) is planar and triangulated.
    Then \(G\) can be uniquely reconstructed from \(T\) if \(n \geq 7\).
\end{reptheorem}

To show this result, we first show that unique reconstruction of such graphs is possible if they are 4-connected.
We then show the case where the graph is \emph{not} 4-connected reduces to the 4-connected case.

\section{Reconstruction from connected $k$-sets}
\label{sec:extensions}
For $k\geq 2$ and a graph $G=(V,E)$, we define the \emph{connected $k$-sets} of $G$ as the set $\{X\subseteq V \mid |X|=k \text{ and }G[X]\text{ is connected}\}$.
We will denote the set of neighbours of a vertex $v$ by $N(v)$.
\begin{observation}
\label{obs:monotonicity}
For $k'\geq k\geq 2$, the connected $k'$-sets of a graph are determined by the connected $k$-sets.
\end{observation}
Indeed, a $(k+1)$-set $X=\{x_1,\dots,x_{k+1}\}\subseteq V$ induces a connected subgraph of $G$ if and only if for some $y,z\in X$, both $G[X\setminus\{y\}]$ and $G[X \setminus\{z\}]$ are connected. 

Given a class $\mathcal{C}$ of graphs, we can consider the function $k(n)$, where for any integer $n \geq 1$, we define $k(n)$ to be the largest integer $k\geq 2$ such that all (labelled) $n$-vertex graphs in $\mathcal{C}$ have a different collection of connected $k$-sets. By Observation \ref{obs:monotonicity}, asking for the largest such $k$ is a sensible question: reconstruction becomes more difficult as $k$ increases. We will always assume that we only have to differentiate the graph from other graphs in the graph class, and remark that often the \emph{recognition problem} (is $G\in \mathcal{C}$?) cannot be solved even from the connected triples.

First, we give an analogue of Theorem \ref{thm:tree-reconstruction}. The proof is given in the appendix.
\begin{reptheorem}{treethreshold}
If it is known that the input graph is a tree, then the threshold for reconstructing trees is at $\lceil n/2\rceil $: we can reconstruct an $n$-vertex tree from the connected $k$-sets if $k\leq \lceil n/2\rceil $ and we cannot reconstruct the order of the vertices in an $n$-vertex path if $k\geq \lceil n/2\rceil +1$.
\end{reptheorem}
Using the theorem above, we give examples in the appendix (Proposition \ref{prop:infmany}) showing that for every value of $k\geq 2$, there are infinitely many graphs that are determined by their connected $k$-sets but not by their connected $(k+1)$-sets. 

We next show that a threshold near $n/2$ that we saw above for trees, holds for `almost every $n$-vertex graph'.
The Erd\H{o}s-Renyi random graph
$G\sim G(n,\frac12)$ has $n$ vertices and each edge is present with probability $\frac12$, independently of the other edges. This yields the uniform distribution over the collection of (labelled) graphs on $n$ vertices. 
If something holds for the random graph with high probability (that is, with a probability that tends to $1$ as $n\to \infty$), then it holds for `almost every graph' in some sense. The random graph is also interesting since it is extremal for many problems. More information can be found in e.g. \cite{bela1998chapter,frieze2015,janson2011random}.

We say an $n$-vertex graph $G=(V,E)$ is \emph{random-like} if the following three properties hold (with $\log$ of base $2$).
\begin{enumerate}
    \item For every vertex $v\in V$,
\[
n/2-3 \sqrt{n\log n}\leq |N(v)|\leq n/2+3 \sqrt{n\log n}.
\]
\item For every pair of distinct vertices $v,w\in V$,
\[
n/4- 3 \sqrt{n\log n}\leq |N(v)\cap (V\setminus N(w))|\leq  n/4+ 3 \sqrt{n\log n}.
\]
\item There are no disjoint subsets $A,B\subseteq V$ with $|A|,|B|\geq 2\log n$ such that there are no edges between a vertex in $A$ and a vertex in $B$.
\end{enumerate}
\begin{replemma}{randomgraphproperties}
For $G\sim G(n,\frac12)$, with high probability $G$ is random-like.
\end{replemma}
The claimed properties of the random graph are well-known, but we added a proof in Appendix \ref{app2} for the convenience of the reader.
\begin{theorem}
\label{thm:reconstr_random}
For all sufficiently large $n$, any $n$-vertex graph $G$ that is random-like 
can be reconstructed from the set of connected $k$-sets for $2\leq k\leq \frac12 n-4 \sqrt{n \log n}$ in time $O(n^{k+1})$. On the other hand, $G[S]$ is connected for all subsets $S$ of size at least $\frac12 n +4 \sqrt{n \log n}$. 
\end{theorem}
In particular, for `almost every graph' (combining Lemma \ref{randomgraphproperties} and Theorem \ref{thm:reconstr_random}), the connectivity of $k$-tuples for $k\geq \frac12 n +4 \sqrt{n \log n}$ gives no information whatsoever, whereas for $k\leq \frac12 n -4\sqrt{n\log n}$ it completely determines the graph.
\begin{proof}[of Theorem \ref{thm:reconstr_random}]
Let  $K$ be the set of connected $k$-sets.

We first prove the second part of the statement. Let $k\geq \frac12n + 4 \sqrt{n\log n}$ be an integer and let $S$ be a subset of $V$ of size at least $k$. Consider two vertices $u,v \in S$. We will prove that $u$ and $v$ are in the same connected component of $G[S]$. By the first random-like property, there are at most $\frac12n + 3 \sqrt{n\log n}$ vertices non-adjacent to $u$, which implies that $u$ has at least $\sqrt{n \log n}-1\geq 2\log n$ neighbours in $S$. Note that, for the same reason, this is also true for $v$. Therefore, we can apply the third random-like property on $A=N(u) \cap S$ and $B=N(v) \cap S$ to ensure that there is an edge between the two sets. We conclude that there must exist a path from $u$ to $v$. 

Let us now prove the first part of the statement. Let $2\leq k \leq \lfloor \frac12n - 4 \sqrt{n\log n} \rfloor$. Let $u \in V(G)$.
We claim that the set of vertices $V\setminus N[v]$ that are not adjacent or equal to $u$, is the largest set $S$ such that $G[S]$ is connected and $G[S \cup \{u\}]$ is not. Note that the two conditions directly imply that $S \subseteq V\setminus N[u]$. To prove equality, it is sufficient to prove that $G[V \setminus N[u]]$ is connected. Consider two vertices $v,w\in G[V \setminus N[u]]$. By the second random-like property, the sets $A= N(v) \cap (V\setminus N[u])$ and $B=N(w) \cap (V\setminus N[u])$ have size at least $n/4 - 3 \sqrt{n\log n}$.
By the third property of random-like, there is therefore an edge between  $A$ and~$B$. This proves that there is a path between $v$ and $w$ using vertices in $V \setminus N[u]$, and so $G[V \setminus N[v]]$ is connected.

For each vertex $u$, the set of vertices it is not adjacent to can now be found by finding the largest set $S$ such that $G[S]$ is connected and $G[S \cup \{u\}]$ is not. Since any such $S$ is a subset of $V\setminus N[u]$, there is a unique maximal (and unique maximum) such $S$. We now give the $O(n^{k+1})$ time algorithm for this. 

We begin by constructing a data structure like the one described by Fredman et al.~\cite{Fredman1982}, which allows us to query the required $k$-sets in $O(1)$. This takes deterministic time of $O(|K|\log |K|)=O(n^k k)=O(n^{k+1})$. For a vertex $v$, we give an algorithm to reconstruct the neighbourhood of $v$ in time $O(n^{k+1})$.
\begin{enumerate}
    \item We first run over the subsets $S$ of size $k$ until we find one for which $G[S]$ is connected but $G[S\cup \{v\}]$ is not. This can be done in time $O(kn^{k})$: if $G[S]$ is connected, then $G[S\cup \{v\}]$ is disconnected if and only if $G[S\setminus\{s\} \cup \{v\}]$ is disconnected for all $s\in S$.
    \item For each vertex $w\in V\setminus(S\cup \{v\})$, we check whether there is a subset $U\subseteq S$ of size $k-1$ for which $U\cup \{w\}$ is connected, and whether for each subset $U'\subseteq S$ of size $k-2$, $U'\cup \{w,v\}$ is not connected. If both are true for the vertex $w$, then $G[S\cup \{w\}]$ is connected and $G[S\cup \{w,v\}]$ is not connected, so we add $w$ to $S$ and repeat this step.
    \item If no vertex can be added anymore, we stop and output $V\setminus (S\cup \{v\})$ as the set of neighbours of $v$.
\end{enumerate} 
We repeat step 2 at most $n$ times, and each time we try at most $n$ vertices as potential $w$ and run over subsets of size at most $k-1$. Hence, this part runs in time $O(n^{k+1})$.
We repeat the algorithm above $n$ times (once per vertex) in order to reconstruct all edges.
\qed \end{proof}
We prove the following analogue to Theorem \ref{thm:triangle-free} in Appendix \ref{app2}.
\begin{reptheorem}{biggirth}
Let $k \geq 4$ be an integer. Every graph on at least $2k -1$ vertices that is known to have no cycles of length at most $k$ is determined by its connected $k$-sets.
\end{reptheorem}

\section{Conclusion}
We have presented a new model of uncertainty in graphs, in which we only receive all triples of vertices that form a connected induced subgraph.
In a way, this is the simplest model of combinatorial indeterminacy in graphs.
We have studied some basic properties of this model, and provided an algorithm for finding a graph consistent with the given indeterminacies.
We also proved that trees, triangle-free graphs and various other families of graphs are determined by the connected triples, although we need to know the family the sought graph belongs to. In order to obtain a full characterisation, it is natural to put conditions on the way a triangle may connect to the rest of the graph, for instance, it is not too difficult to recognise that $a,b,c$ induces a triangle if at least two of $a,b,c$ have private neighbour, whereas it is impossible to distinguish whether $a,b,c$ induce a triangle or a path if all three vertices have the same neighbours outside of the triangle. We leave this open for future work.
 
Similar to what has been done for graph reconstruction (see e.g. \cite{LS16}), another natural direction is to loosen the objective of reconstruction, and to see if rather than determining the (labelled) graph, we can recover some graph property such as the number of edges or the diameter. A natural question is also how many connected triples are required (when given a `subcollection', similar to \cite{BBF10,GroenlandGuggiarScott2020,M92}, or when we may perform adaptive queries, as in~\cite{kannan2018reconstruction}). 

We gave various results in an extension of our model to larger \(k\)-sets, including trees and random graphs.
There are several other logical extensions to the concept of reconstructing a graph from connected triples.
We could define a \((k, \ell)\)-representation \(T\) to contain all \(k\)-sets that are connected and contain at least \(\ell\) edges.
The definition of connected triples would then be a \((3, 2)\)-representation.
Note that in this case some vertices may not appear in \(T\), or \(T\) might even be empty altogether (e.g. for trees when \(\ell \geq k\)).
Another natural extension would be to specify the edge count for each $k$-set, but this gives too much information even when $k=n-2$: the existence of any edge $\{u,v\}$ can be determined from the number of edges among vertices in the four sets $V,V\setminus\{u\},V\setminus\{v\}$ and $V\setminus\{u,v\}$.

Some interesting algorithmic questions remain open as well. In particular, we presented an efficient algorithm to specify whether a collection of connected $3$-sets uniquely determines a graph, but do not know how to solve this efficiently for larger values of $k$. Is the following decision problem solvable in polynomial time: given a graph $G$ and an integer $k$, is $G$ determined by its collection of $k$-tuples? We note that when $k$ is fixed to 4, membership in coNP is clear (just give another graph with the same connected $4$-sets) whereas even NP-membership is unclear. 

Finally, a natural question is whether the running time of $O(n\cdot |T|)$ for finding a consistent graph with a set of connected triples (Lemma \ref{lem:algorithm}) can be improved to $O(|T|)$ (or expected time $O(|T|)$).

\bibliographystyle{splncs04}
\bibliography{bibliography}

\appendix

\section{Omitted proofs for connected triples}
In this appendix, we give the proofs that were omitted from Section \ref{sec:unique-graphs}.

\subsection{Reconstructing 2-connected outerplanar graphs}\label{subsec:outerplanar}
We show that 2-connected outerplanar graphs \(G\) on at least six vertices are determined by their set of connected triples \(T\).

We begin by describing some properties of planar and outerplanar graphs. Wagner's Theorem \cite{wagner} states that a graph is planar if and only if it does not contain the complete graph $K_5$ or the complete bipartite graph $K_{3,3}$ as a minor. 
It is an easy consequence that outerplanar graphs are exactly the graphs that do not contain $K_4$ or $K_{2,3}$ as a minor. 
An easy structural analysis shows that 
every outerplanar graph has a vertex of degree at most two. Since this fact is well-known we omit the proof, but it can be found in e.g. \cite{westPlanarGraph} (as Proposition 6.1.19). 

Our approach is similar to the one for trees: we show that we can identify a vertex of degree two, and remove it from the graph by merging it with one of its neighbours.
We keep doing this until we have a graph with six vertices, which can be distinguished by the number of triples in which each vertex and edge occurs.
We first observe that we can recognise vertices of degree two. 
We will need the following lemma which we obtain by applying Menger's well-known theorem  (see \cite{westPlanarGraph} for further background).

\begin{lemma}\label{lem:fourcycle-deg-2}
Let $G$ be a 2-connected outerplanar graph on at least five vertices.
Suppose $C$ is a $4$-cycle in $G$.
Suppose $a,b \in C$ each have no neighbours in $V(G) \setminus V(C)$.
Then $a,b$ must be consecutive in $C$.
\end{lemma}
\begin{proof}
    Suppose that $a, b$ are not consecutive in $C$.
    Denote the other two vertices of $C$ by $c, d$.
    Since $|V(G)| \geq 5$ there is some vertex $r \in V(G) \setminus C$.
    By Menger's theorem, there are paths $P$ and $Q$ starting at $r$ and ending in $\{a,b,c,d\}$ such that $P$ and $Q$ are vertex disjoint (except for $r$).
    Then $P \cup Q$ is a path between $b$ and $d$.
    But then $C \cup P \cup Q$ is a $K_{2,3}$-minor, contradicting the fact that $G$ is outerplanar. 
\qed \end{proof}

\begin{lemma}\label{lem:degree-2-vertex}
    Let $G$ be a 2-connected outerplanar graph on at least six vertices.
    A vertex \(v \in V(G) \) has degree two, and its neighbours are \(w_1\) and \(w_2\), if and only if all triples \(v\) appears in also contain \(w_1\) or \(w_2\), and there is no other vertex $v'$ for which all triples $v'$ appears in also contain $w_1$ or $w_2$.
\end{lemma}

\begin{proof}
    Suppose $v$ has exactly two neighbours, $w_1$ and $w_2$.
    Then, clearly $v$ will only appear in triples with these vertices.
    Suppose for a contradiction that there is some $v' \in V(G) \setminus \{w_1, w_2, v\}$ such that every triple containing $v'$ also contains $w_1$ or $w_2$. It follows that $v'$ has at most one neighbour in $V(G) \setminus \{w_1, w_2\}$.

    Suppose $v'$ is adjacent to both $w_1$ and $w_2$.
    Then $v' \dd w_1 \dd v \dd w_2 \dd v'$ is a four cycle $C$ in $G$.
    So by Lemma \ref{lem:fourcycle-deg-2}, $v'$ must have a neighbour $r \in V(G) \setminus V(C)$.
    Moreover, $N(r) \subseteq V(C)$ because otherwise $v'$ occurs in some triple not containing $w_1, w_2$.
    Hence, $G \setminus r$ is $2$-connected.
    Then, by applying Lemma \ref{lem:fourcycle-deg-2} to $v,v'$ and $C$ in $G \setminus r$, we obtain a contradiction.
    Thus, we may assume that $v'$ is non-adjacent to $w_2$.

    $G$ is 2-connected, so $v'$ has degree at least two.
    It follows that $v'$ is adjacent to $w_1$.
    Since $G$ is 2-connected, there is a shortest $v'w_2$-path $P_2$ in $G \setminus w_1$.
    Let $z$ be the neighbour of $v'$ in $P_2$. Then $z \neq w_2$.
    Then $N(z) \subseteq \{v', w_1, w_2\}$ because otherwise $v'$ would be in a triple not containing $w_1$ or $w_2$.
    Thus, $v \dd z \dd w_2$ is a path of $G$.
    Hence, the graph $G'$ obtained by deleting $z_2$ and adding the edge $w_1w_2$ is 2-connected.
    Since outerplaner graphs are minor-closed, $G'$ is  outerplanar.
    By Lemma \ref{lem:fourcycle-deg-2} applied to the cycle $v \dd w_1 \dd v' \dd w_2 \dd v$ in $G'$ we obtain a contradiction since $v, v'$ both have no neighbour in $V(G') \setminus \{v, w_1, v', w_2,\}$. 
    
    Conversely, suppose all triples \(v\) appears in also contain \(w_1\) or \(w_2\), and there is no other vertex $v'$ for which all triples $v'$ appears in also contain $w_1$ or $w_2$.
    Suppose $v$ has some neighbour $z \neq w_1, w_2$.
    Since $v$ only appears in triples with $w_1$ and $w_2$ it follows that $N(z) \subseteq \{w_1, w_2, v\}$.
    But then $z$ only appears in triples with $w_1$ and $w_2$, a contradiction.
    Hence, $w_1$ and $w_2$ are the only possible neighbours of $z$.
    Since $G$ is 2-connected, $v$ has degree at least two.
    So $v$ has degree two and it is adjacent to $w_1$ and $w_2$.
\qed \end{proof}

\begin{figure}[t]
    \centering
    \includegraphics{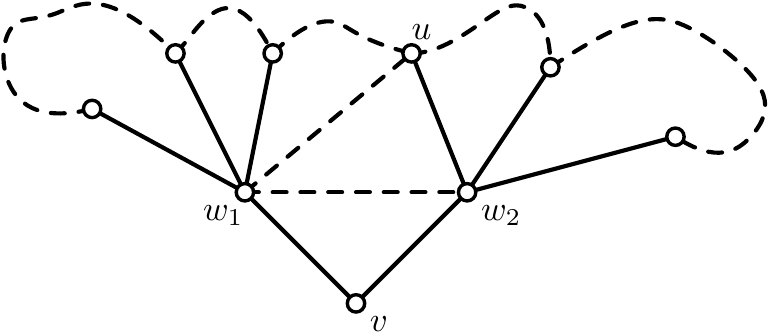}
    \caption{The local neighbourhood of a degree two vertex in a 2-connected outerplanar graph on at least 6 vertices.
    The dashed edges \(w_1w_2\) and \(w_1u\) may or may not be present.}
    \label{fig:degree-two-vertex}
\end{figure}

We will use Lemma \ref{lem:degree-2-vertex} to reduce the problem of recognising a 2-connected outerplanar graph on at least \emph{seven} vertices to recognising a 2-connected outerplanar graph on \emph{six} vertices.
This significantly, simplifies the situation as there are only nine distinct 2-connected outerplanar graphs on six vertices up to isomporphism (see Figure \ref{fig:6-vertex-outerplanar}).

\begin{figure}[t]
    \centering
    \includegraphics{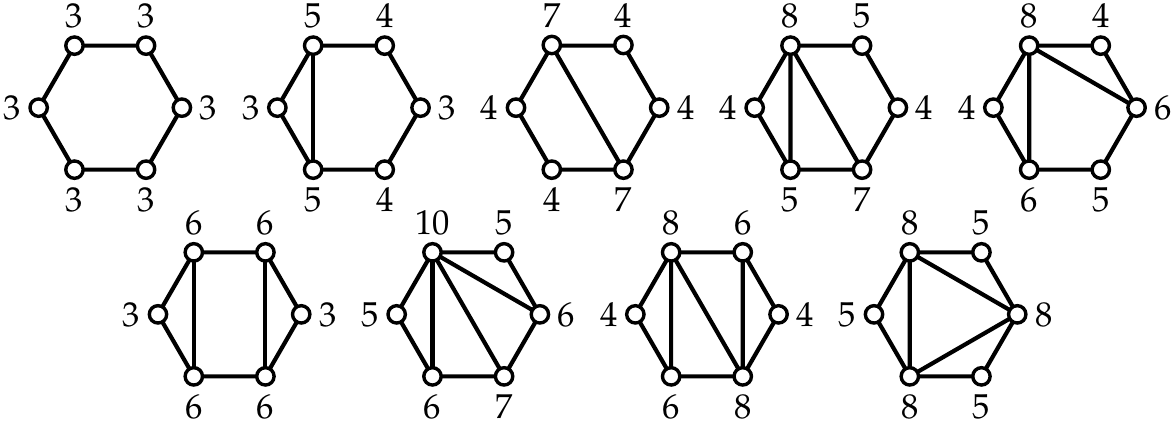}
    \caption{All base cases for a 2-connected outerplanar graph of six vertices, up to symmetries. Each vertex is labelled by the number of triples it occurs in.
    Note that each case can be identified by the number of triples each vertex occurs in.}
    \label{fig:6-vertex-outerplanar}
\end{figure}

\begin{observation}\label{obs:six}
Every graph on six vertices that is known to be 2-connected and outerplanar can be reconstructed. 
\end{observation}

For each vertex $v$ in a graph $G$ let $\occ(v)$ denote the number of times $v$ occurs in a triple in $G$ and let $\occ(G)$ denote the unordered list of values of $\occ(w)$ for each $w \in V(G)$.
Each 2-connected outerplanar graph $G$ on six vertices has a unique list $\occ(G)$; see \cref{fig:6-vertex-outerplanar} for an illustration.
Then, by this fact, Lemma \ref{lem:degree-2-vertex} and some easy extra analysis of the triples we can reconstruct any graph on six vertices that is known to be 2-connected and outerplanar.

\repeattheorem{twoconnected}

Note that Theorem \ref{twoconnected} is tight in the sense that it does not hold for $n = 4,5$ and holds trivially when $n = 3$.
For $n = 5$, consider the graph $H$ that is obtained from a $P_4$ by adding a new vertex complete to everything. We can determine whether a graph $G$ is isomorphic to $H$ from its triples and the knowledge that $G$ is 2-connected and outerplanar, but we cannot determine the order of the vertices of the $P_4$ (See Figures \ref{fig:label-ambiguity} and \ref{fig:ambiguous-classes}).
For $n = 4$, note that we cannot distinguish between $C_4$ and the graph obtained from $C_4$ by adding a chord since all sets of three vertices will be connected for both graphs.

\begin{proof}[Proof of \cref{twoconnected}]
    By Observation \ref{obs:six}, the theorem holds when $n = 6$.
    Suppose $G$ is outerplanar and 2-connected and that $G$ has at least seven vertices and suppose the theorem holds for $n = 6, 7, \dots, |V(G) -1$. 
    
    Since $G$ is outerplanar and  2-connected, $G$ has a vertex $v$ of degree two. Suppose $v$ has a neighbour $u$ of degree two.
    We can recognise $u, v$ and their neighbours by Lemma \ref{lem:degree-2-vertex}. We record this information.
    Contract the edge $\{u, v\}$ to a single vertex $x$ and call the resulting graph $G'$.
    Then $G'$ is $2$-connected and outerplanar.
    We obtain the list of triples for $G'$ from the list of triples of $G$ by removing all triples that contain both $u$ and $v$ and then replacing any remaining occurrence of $u, v$ with $x$.
    By induction, we can reconstruct $G'$, and therefore $G$, from its list of triples.

    Hence, we may assume the neighbours of $v$ are both of degree at least three. Let $w_1, w_2$ denote the neighbours of $v$.
    By Lemma \ref{lem:degree-2-vertex}, we can identify $v, w_1, w_2$.
    We also know all neighbours of $w_1, w_2$: A vertex $x \in V(G) \setminus \{w_1, w_2, v\}$ is adjacent to $w_1$ ($w_2$) if and only if $vxw_1$ (resp. $vxw_2$) is a triple. 
    This does not tell us if \(w_1\) and \(w_2\) are neighbours, as we have the triple \(uw_1w_2\) either way.
    However, we do know that, as $G$ is outerplanar, \(w_1\) and \(w_2\) can share at most two neighbours, one of which is \(v\) (see \cref{fig:degree-two-vertex} for an illustration).
    This means both \(w_1\) and \(w_2\) have at least one neighbour that is not a neighbour of the other vertex, as they are both degree at least three.
    Hence, $w_1$ and $w_2$ are non-adjacent if and only if there is some vertex $x$ in $G$ such that $w_1w_2x \not \in T$ but $vw_1x \in T$ or $vw_2x \in T$.
    Thus, we can determine if $w_1$ and $w_2$ are neighbours.

    If the edge \(w_1w_2\) exists, then $G \setminus v$ is still 2-connected.
    Therefore, we can remove \(v\) and all triples.
    If the edge $w_1w_2$ does not exist, we can add it and add all the triples \(w_1w_2x\) for each neighbour \(x\) of \(w_1\) or \(w_2\), then remove \(v\) and all triples in which it occurs.
    In either case we obtain a 2-connected outerplanar graph $G'$ with one fewer vertex and its list of triples, so we can reconstruct $G'$ by induction.  
    Moreover, in both cases we can reconstruct $G$ from $G'$ since we know $v, w_1, w_2$ and their incident edges.
\qed \end{proof}

\input{triangulated-planar}

\section{Proofs for reconstruction from connected $k$-sets}
\label{app2}
In this appendix we give the proofs about reconstruction from connected $k$-sets that were omitted from the text. We repeat the results for convenience of the reader.

\repeattheorem{treethreshold}
\begin{proof}
For the first claim, by monotonicity we may assume that $k = \lceil n/2 \rceil$. We first find the set $L$ of all leaves by checking for all $v \in V$ whether $G[V\setminus\{v\}]$ is connected. Note that since $G$ is a tree, this is only the case if $v$ is a leaf.

Next we check for all $v \in V$:
\begin{enumerate}
    \item for all $L' \subseteq L$ of size at least $\lceil n/2\rceil -1$ whether $G[L' \cup \{v\}]$ is connected, and
    \item for all $L' \subseteq L$ of size at most $\lceil n/2\rceil -2$ whether $G[V\setminus (L' \cup \{v\})]$ is connected.
\end{enumerate} 
The size constraints put on $L'$ allow us to perform the queries, as 
\[
n-(\lceil n/2\rceil -1) = n-\lfloor n/2\rfloor)= \lceil n/2\rceil.
\]
All $L'$ giving a positive response for the first question (whether $G[L' \cup \{v\}]$ is connected) are contained in $N(v)\cap L$. 
Indeed, they are contained in $L$ by assumption and a leaf is only part of a connected subgraph if its `parent' is also present.
So if $|N(v) \cap L| \geq  \lfloor n/2\rfloor-1$, then we recover $N(v)\cap L$ as the largest $L'\subseteq L$ such that $G[L' \cup \{v\}]$ is connected.

On the other hand, all $L'\subseteq L$  giving a positive response to the second question (whether $G[V\setminus (L' \cup \{v\})]$ is connected)
must contain $N(v)\cap L$, unless $v$ is the only non-leaf vertex (in which case we can immediately conclude $G$ is a star). Therefore, if $|N(v) \cap L| \leq \lceil n/2 \rceil - 2$ then we recover $N(v) \cap L$ as the smallest set $L'$ for which the response is positive.

Once we have recovered for each vertex what leaves it is incident to, we may continue to find the set of vertices adjacent to those in a similar manner: we see a vertex that has a leaf incident to it as a `second-round' leaf, and assign it as a weight the number of leaves incident to it plus 1. If we include a second-round leaf in a subset, then we include itself and its incident leaves, and otherwise the same argument applies to find `third-round' leaves, et cetera.

For the second claim let $P_n = ([n], \{\{i, i+1\} : i \in [n-1]\}$ be an $n$-vertex path. If $n$ is even, then any connected set of $n/2 + 1$ vertices must include the vertices $n/2$ and $n/2 + 1$. This implies that we cannot distinguish between the current path and a path where these two vertices are swapped. Similarly, if $n$ is odd, any connected set of $\lceil n/2 \rceil + 1$ vertices must contain the vertices $\lceil n/2 \rceil - 1$, $\lceil n/2 \rceil$ and $\lceil n/2 \rceil + 1$. This again implies that we cannot distinguish between these three vertices.
\qed \end{proof}
The result above can be used to show that for each $k$, infinitely many graphs have  `reconstructibility threshold' $k$ in the following sense.
\begin{proposition}
\label{prop:infmany}
For each value of $k\geq 2$, there exist infinitely many graphs that are determined by their connected $k$-sets but not by their connected $(k+1)$-sets
\end{proposition}
\begin{proof}
    The claim is clear for $k=2$. Let $k\geq 3$. Let $P$ be a path on $2k$ vertices and let $P'$ be a path on $n\geq 2k$ vertices. We add all edges between vertices of $P$ and $P'$ and let $G$ be the resulting graph. If a set $S$ intersects both $P$ and $P'$, then $G[S]$ is always connected, so we can only reconstruct $G$ from its connected $\ell$-set if we can reconstruct $P$ from its own connected $\ell$-sets. By Theorem \ref{thm:tree-reconstruction}, we can reconstruct $P$ from the connected $\ell$-sets of $P$ if and only if $\ell\leq k$. This means we may confuse $G$ with other graphs when given the connected $(k+1)$-sets. On the other hand, given the connected $k$-sets we can find the (non)-edges of $P$ and $P'$, and then discover they are fully connected since $G[\{p\}\cup I]$ is connected for any $p\in P$ and $I\subseteq P'$ independent set of size $k-1$.
\qed \end{proof}

We next prove the properties about the random graph.
\repeatlemma{randomgraphproperties}
(We could improve on the constants in definition of random-like, but we chose our constants for readability purposes.)
For the proof, we use Chernoff bounds as stated below.
\begin{lemma}[Chernoff bound]
\label{lem:chernoff}
Let $X_1,X_2,\dots,X_n$ be independent random variables taking values in $\{0,1\}$ with $\P(X_i=1)=p$. Let $X=\sum_{i=1}^n X_i$ and $\mu=\E[X]=pn$. Then, for any $\epsilon>0$,
\[
\P(|X-\mu|\geq \epsilon\mu)\leq 2\exp(-\epsilon^2\mu/3).
\]
\end{lemma}
\begin{proof}[Proof of Lemma \ref{randomgraphproperties}]
Fix a vertex $v\in [n]$ and let $u\in [n]\setminus\{v\}$ be another vertex. Let $X_u$ be the indicator function of the event that the edge $\{u,v\}$ is present in $G\sim G(n,\frac12)$. Then $\P(X_u=1)=\frac12$ and $|N(v)|=\sum_{u\neq v}X_u$ is a sum of independent random variables. Hence, by a Lemma \ref{lem:chernoff} (Chernoff bound) with $\mu=(n-1)/2$, $\epsilon=5\sqrt{\log n/n}$ and $X=|N(v)|$, we find for $n$ sufficiently large that
\begin{align*}
\P(||N(v)|- n/2|\geq 3\sqrt{n\log n})&\leq \P(||N(v)|- (n-1)/2|\geq 2.5\sqrt{n\log n})\\
&= \P(|X- \mu|\geq \mu \frac{n}{n-1} 5\sqrt{\log n/n})\\
&\leq \P(|X- \mu|\geq \epsilon \mu)\\
&\leq 2\exp(-\epsilon^2\mu/3)\\
&= 2\exp(-\frac{25}{6}\log n\frac{n-1}{n})\leq 2n^{-2.5}.
\end{align*}
By a union bound over vertices, the probability that the first property of random-like holds for $G(n,\frac12)$ converges to 1 as $n\to \infty$. 

To obtain the second statement of the definition of random-like, let $v,w\in V$ be distinct vertices, and for $u\in V\setminus\{v,w\}$, let $X_u$ denote the event that $\{u,v\}$ is an edge and $\{u,w\}$ is not an edge. Again, the $X_u$ are independent as they concern distinct edges and $\P(X_u=1)=\frac12\frac12=\frac14$, so $X=|N(v)\cap (V\setminus N(w)|=\sum_{u\in V\setminus\{v,w\}}X_u$ satisfies $\mu=\E[X]=(n-2)/4$.  By a similar calculation, we find that
\[
\P(||N(v)\cap (V\setminus N(w))|-n/4|\geq 3\sqrt{n\log n})\leq 2n^{-3}.
\]
By a union bound over the pairs of vertices, the probability that the second property of random-like holds for $G(n,\frac12)$ converges to 1 as $n\to \infty$. 

It remains to show that the third property holds with high probability.
Given two disjoint sets $A,B \subseteq V$ such that $|A|,|B| \geq 2\log(n)$, we denote by $E_{A,B}$ the event that there is no edge with one endpoint in $A$ and the other in $B$. Note that there are exactly $|A|\cdot|B|$ of such potential edges. We compute
\[
\P(E_{A,B}) = \left(\frac12\right)^{|A|\cdot|B|} \leq \left(\frac12\right)^{4\log^2 n }.
\] 
In order to prove the third property it suffices to prove that for all sets $A,B$ of size exactly $2 \log n$ the event $E_{A,B}$ does not happen (with high probability). As there are less than $\binom{n}{2\log n}^2$ pairs of sets of this size, we find
\[
\P \left(\bigcup_{A,B \in \binom{[n]}{2\log n}} E_{A,B} \right) \leq \binom{n}{2\log n}^2 \left(\frac12\right)^{4\log^2 n} \leq \left(\frac{en}{2\log n}\right)^{4\log n} n^{-4\log n} \leq \left(\frac{e}{2\log n}\right)^{4 \log n}.
\]
The third inequality comes from the well-known bound $\binom{n}{k} \leq \left(\frac{en}{k}\right)^k$. Note that the last term converges to 0 when $n \rightarrow \infty$.

This shows all three properties hold with high probability.
\qed \end{proof}

Next, we turn back to the deterministic setting and prove Theorem \ref{biggirth}, an analog of Theorem \ref{thm:triangle-free} for graphs of high girth, which we restate below.

In this section we define the \emph{girth} of a graph $G$ to be the length of a shortest cycle in $G$. 
For brevity, if $G$ is acyclic we say that $G$ has infinite girth. 

\repeattheorem{biggirth}

The proof of Theorem \ref{biggirth} is similar to that of Theorem \ref{thm:triangle-free}.
For the remainder of this section we take $k \geq 4$ to be a fixed integer.
We begin by showing how to reconstruct the induced cycles from the connected $k$-sets of graphs that are known to have girth at least $k+1$.

\begin{lemma}\label{lem:cycle-unlabelled}
Let $G$ be a graph without cycles of length at most $g$ for some $g \geq 3$.
For any subset $A\subseteq V$ of size $g+1$, $G[A]$ induces a cycle of length $g+1$ if and only if $G[A\setminus \{a\}]$ is connected for all $a\in A$.
\end{lemma}
\begin{proof}
Suppose that $G[A\setminus\{a\}]$ is connected for each $a\in A$. 
Let $T$ denote $G \setminus a$ for some fixed $a \in A$.
Since $G$ has no cycles of length $3,\dots,k$, it follows that $T$ is a tree.
We claim that in fact $T$ must always be a path, which proves the result since $a\in A$ was chosen arbitrarily.

If $T$ is not a path, then it has some vertex $t\in T$ of degree at least 3.
By assumption $G[A\setminus \{t\}]$ is also connected, but $G[A\setminus \{t,a\}]$ has at least three connected components $C_1, C_2, \dots, C_\ell$.
Note that both $t$ and $a$ must be adjacent to each connected component of $G[A \setminus \{t,a\}]$.
So $G[C_1\cup C_2\cup \{t,a\}]$ contains a cycle.
But it also has at most $g$ vertices since we exclude all vertices of $C_3, \dots, C_\ell$.
This is a contradiction. 
\qed \end{proof}
Applying the lemma above, we can recognise the $(k+1)$-cycles, that is, we can reconstruct collection \[
\{\{s_1,\dots,s_{k+1}\}\subseteq V:G[\{s_1,\dots,s_{k+1}\}]\text{ induces a $(k+1)$-cycle}\}.
\]
Let $S=\{s_1,\dots,s_{k+2}\}$ denote a set of $k+2$ vertices for which no subset of size $k+1$ appears in the collection above. Then we know $G[S]$ has no cycles of length $3,\dots,k+1$, applying Observation \ref{obs:monotonicity} we can also recognise the $(k+2)$-cycles. Continuing this way, we  may recognise all subsets of vertices of $G$ that induce a cycle from the connected $k$-sets, if we know $G$ to have no cycles of length at most $k$. In particular, we obtain the following corollary to Lemma \ref{lem:cycle-unlabelled}.
\begin{corollary}\label{corr:tree}
Let $G$ be a graph which we know has no cycles of length at most $k$. Then we can recognise all subsets that induce a cycle from the connected $k$-sets, and can in particular recognise whether the graph is a tree.
\end{corollary}

We extend Corollary \ref{corr:tree} by showing we can \emph{reconstruct} the order of the labelled vertices of the cycles as well. We begin by proving the result for cycles of length at least $k+2$.

\begin{lemma}\label{lem:cyclelabel}
Let $G$ be a graph that is known to have girth at least $k +1$.
Then, we can reconstruct every cycle in $G$ of length at least $k+2$ from the connected $k$-sets of $G$.
\end{lemma}
\begin{proof}
    By Corollary \ref{corr:tree} we can identify every set of vertices that induces a cycle in $G$.
    Suppose $V(C)$ is a set of vertices that induces a cycle of $C$ of length at least $k+2$.
    Note for each $x,y \in V(C)$ the vertices $x$ and $y$ are adjacent if and only if the $C \setminus \{x,y\}$ is connected, and that $|V(C) \setminus \{x,y\}| \geq k$. By monotonicty (Observation \ref{obs:monotonicity}), the connected $k$-sets of $G$ determine whether $C \setminus \{x,y\}$ is connected, and so we can discover the order of the vertices on the cycle.
\qed \end{proof}
Since every subset of $k$ vertices of a cycle of length $k+1$ is connected, we cannot apply the same argument to reconstruct the cycles of length $k+1$. However, since we only consider graphs which are known to have at least $k +2$ vertices, we can obtain additional information about the order of vertices cycle of length $k+1$ by looking at their relationships with other vertices in the graph.
\begin{lemma}\label{lem:cycles-and-labels}
    Let $G$ be a connected graph that is known to have girth at least $k+1$. Suppose $C$ is an induced cycle in $G$. Then, we can reconstruct $C$ from the connected $k$-sets of $G$.
\end{lemma}
\begin{proof}
    By Lemma \ref{lem:cyclelabel} we may assume $C$ has length $k+1$.
    By Corollary \ref{corr:tree}, we can recognise the vertex set of $C$.
    A vertex $w \not \in V(C)$ is a neighbour of a vertex in $V(C)$ if and only if there is some connected $k$-set containing $w$ and $k-1$ vertices of $C$.
    Hence, we can identify the set $S$ of vertices in $V(G) \setminus V(C)$ that have at least one neighbour in $V(C)$.
    Since $G$ is connected, $S \neq \emptyset$.
    
    Let $s \in S$ be arbitrary. Since $G$ has no cycles of length at most $k$, it follows $s$ has exactly one neighbour in $V(C)$. We can identify the unique neighbour $v\in V(C)$ of $s$ as the only vertex that appears in all subsets $M\subseteq V(C)$ of size $k-1$ for which $\{s\}\cup M$ is connected.
    
    We use this to determine the order of the labels of the vertices in $G$ in a similar fashion to Lemma \ref{lem:cyclelabel}.
    Distinct vertices $x,y \in V(C)$ are adjacent if and only if $C \setminus \{x,y\}$ induces a connected graph. Note $V(C) \setminus \{x,y\}$ contains only $k-1$ vertices, and since $v$ is the unique neighbour of $s$ in $V(C)$, the following statement holds. Two distinct vertices $x, y \in V(C) \setminus \{v \}$ are adjacent if and only if $C \cup \{s \} \setminus \{x,y\}$ is a connected $k$-set.
    Thus, we can reconstruct the (non)-edges between vertices in $C \setminus \{v \}$. The graph induced by
    $C \setminus \{v\}$ is a path and the two degree-one vertices of $C \setminus \{v\}$ are the two neighbours of $v$ in $V(C)$.
    It follows that we can reconstruct $C$.
\qed \end{proof}

We now consider how to reconstruct the other edges (which are not in a cycle) of a graph with a smallest cycle of size at least $k+1$.
Let $C$ be a cycle and let $T$ be a rooted tree on at least two vertices.
We say the graph obtained by identifying the root of $T$ with a vertex in $C$ is a \emph{tree rooted on a cycle}.

\begin{lemma}\label{lem:tree-rooted-at-cycle}
Let $G$ be a graph that is known to have girth at least $k+1$.
Then if $G$ contains exactly one cycle and if $|V(G)| \geq k+2$ then $G$ can be reconstructed from its connected $k$-sets.
\end{lemma}
\begin{proof}
    Let $G$ be a graph of girth at least $k+1$ on at least $k+2$ vertices. We may assume that $G$ contains exactly one induced cycle $C$, and we can recognise that $G$ has this cycle, as well as reconstruct it, by Lemma \ref{lem:cycles-and-labels}.
    For each $u \in V(G) \setminus V(C)$ we say the \emph{distance} from $u$ to $C$ is the length of the shortest path from $u$ to a vertex in $V(C)$.

    For $v \in V(C)$ we can identify a connected set $S$ of $k-1$ vertices of $C$ containing it.
    Moreover, $u \not \in V(C)$ is adjacent to $v$ if and only if $S \cup \{u\}$ is connected.
    Hence, we can reconstruct the graph induced by $C$ and all vertices at distance one from $C$.
    
    Let $\ell > 1$ be an integer.
    Let $X$ be the set of vertices of distance at least $\ell$ from $C$ in $G$ and suppose $X \neq \emptyset$.
    Suppose by induction (on $\ell$) that we have reconstructed $G \setminus X$.
    Let $u$ be a vertex at distance $\ell -1$ from $C$.
    Using vertices on the shortest path from $u$ to $C$, and vertices from $C$ as needed, we select a connected set $S$ of $k-1$ vertices of $G \setminus X$ for which $u \in S$ is the only vertex at distance $\ell-1$ from $C$.
    Then for each $x \in X$, the vertex $u$ is adjacent to $x$  if and only if $S \cup \{x \}$ is connected. This way we reconstruct the subgraph of $G$ consisting of vertices at distance at most $\ell$, until $X$ is empty and the graph has been reconstructed.
    The lemma follows.
\qed \end{proof}

We are now ready to prove Theorem \ref{biggirth}, which we restate below.
\repeattheorem{biggirth}
\begin{proof}
    Let $G$ be a connected graph on at least $2k-1$ vertices. Suppose $G$ is known to have no cycles of length at most $k$.
    By applying Lemma \ref{lem:cycles-and-labels}, we can recognise all subsets that induce a cycle (together with the order of the vertices on the cycle).
    If $G$ is acyclic then $G$ can be reconstructed from its connected $k$-sets by Theorem \ref{treethreshold}.
    Hence, we may assume $G$ contains some cycle.
    
    Then, we can reconstruct the remaining edges of $G$ as follows.
    Let $U$ denote the union of all vertex sets of induced cycles in $G$. 
    Let $C$ be some induced cycle in $G$.
    Then $C$ is the unique cycle of some connected component of $G[V \setminus U \cup V(C)]$. We can reconstruct that connected component using Lemma \ref{lem:tree-rooted-at-cycle} and record the edges we found. Repeating this for every induced cycle, we record all edges of $G$: indeed, each edge must be connected to at least one of the induced cycles since $G$ is assumed to be connected. 
    This completes the proof. 
\qed \end{proof} 
We stress that we cannot always  \emph{recognise} from the connected $k$-sets whether the graph has a cycle of length $3,\dots,k$, even if the graph is very large. Indeed, if we know three vertices $u,v,w$ all have the same neighbours (not counting $u,v,w$ themselves), and form a connected triple, then we can never distinguish whether they induce a triangle or not. 
\end{document}

%% file: triangulated-planar.tex
\subsection{Reconstructing triangulated planar graphs}
\label{subsec:planar}
A \emph{triangulated planar graph}, also called a \emph{maximal planar graph}, is a planar graph where every face (including the outer face) is a triangle.
By definition triangulated planar graphs on more than three vertices are $3$-connected.
Note that chordal graphs (graphs in which every cycle of more than three vertices has a chord) are sometimes also called triangulated graphs, but this is not the meaning we use here. Every triangulated planar graph is chordal, but the reverse direction is trivially false.

To show that triangulated planar graphs on at least seven vertices can be uniquely reconstructed, 
we take a similar approach as in the proof of \cref{twoconnected}.
We argue that we can recognise all vertices of degree at least five with their neighbours in cyclic order in a graph that is known to be a $4$-connected triangulated planar graph on at least six vertices.
We then argue that we can reconstruct all graphs on at least six vertices that are known to be $4$-connected triangulated planar graphs.
Finally, we will use this result to show we can reconstruct every graph on at least seven vertices that is known to be triangulated and planar.
We begin with some easy observations.

For a connected graph $G$ we say a set of vertices $S \subset V(G)$ satisfying $\emptyset \neq S \neq V(G)$ is a \emph{separator} if $G \setminus S$ is disconnected.
A \emph{separating triangle} is a separator of cardinality three that induces a triangle.

\begin{observation}
    \label{obs:4conn}
    Every separator of cardinality 3 in a triangulated planar graph on at least five vertices is a separating triangle. 
\end{observation}
Indeed, by definition triangulated planar graphs do not have separators of size at most two.
Moreover, minimal separators of triangulated planar graphs must be cliques.
It is a consequence of Euler's formula (given below) that every triangulated planar graph on seven vertices contains a vertex of degree five.
We call the 4-regular planar graph on six vertices the octahedron (See \cref{fig:6-vertex-triangulated-cases} for a depiction).

\begin{theorem}[Euler's formula]\label{thm:euler}
    Let $G$ be a connected plane graph.
    Then $$n - e + f = 2$$ where $n,e,f$ are equal to the number of vertices, edges and faces of $G$, respectively. 
\end{theorem}
\begin{corollary}\label{corr:euler-degree5}
    Every triangulated planar graph with no vertex of degree at least five contains at most six vertices.
    Moreover, the octahedron is the unique triangulated planar graph on six vertices with no vertex of degree five. 
\end{corollary}

For $\ell \geq 4$, we call the graph obtain from a cycle of length $\ell$ by adding a new vertex $v$ and making it adjacent to every other vertex an \emph{$\ell$-wheel}. We call $v$ the \emph{centre} of the $\ell$-wheel. The cycle that does not contain $v$ is called the $C_\ell$ of the $\ell$-wheel.
We will show that for every vertex $v$ of degree $\ell \geq 5$ in a graph that is known to be triangulated, 4-connected, and planar the graph induced by $v$ and its neighbours is an $\ell$-wheel. For brevity, we will denote the set $N(v) \cup \{v\}$ as $N[v]$.

 \begin{lemma}\label{lem:nbd-of-degree-6}
    Suppose $G$ is a triangulated planar graph. Then for every $v \in V(G)$ of degree $\ell \geq 5$, the graph induced on $N[v]$ contains an $\ell$-wheel with centre $v$. Moreover, if $G$ is $4$-connected then the graph induced by $N[v]$ is an $\ell$-wheel.
 \end{lemma}
 \begin{proof}
    Let $v \in V(G)$ have degree at least five.
    Fix an embedding of $G$ and an orientation of the plane.
    Let $w_1, w_2, w_3, w_4, \dots, w_d$ be the neighbours of $v$ in cyclic order according to the embedding of $G[N(v)]$.
    Let $i \in \{1,2, \dots, d \}$.
    Then $w_i \dd v \dd w_{i+1}$ is a path bordering a face of $G$.
    Hence, there is some $w_iw_{i+1}$ path $P$ in $G \setminus v$ such that the union of $P$ and the path $w_i \dd v \dd w_{i+1}$ is a cycle bordering a face of $G$.
    Hence, since $G$ is triangulated $P$ must be a single edge.
    It follows that $w_1 \dd w_2 \dd w_3 \dd \dots \dd w_d \dd w_1$ is a cycle of $G$.
    This proves the first claim.

    Suppose $G$ is $4$-connected. 
     Suppose $w_iw_j$ is an edge for some non-consecutive $i,j \in \{1,2, \dots, d\}$ satisfying $i < j < d$.
    Since $G$ has no separating triangles there is a $w_{i+1}w_{j+1}$-path $Q$ in $G \setminus \{v,w_i, w_j\}$. But then $G[N[v]] \cup Q$ has a $K_5$ minor, a contradiction.
    Hence, $G[N[v]]$ is an $\ell$-wheel.
 \qed \end{proof}


\begin{observation}
    \label{obs:P5C5}
   For each $\ell \geq 5$, $C_\ell$ can be reconstructed.
\end{observation}

The result is trivial, but we include a proof for completeness.

\begin{proof}
Let $\ell \geq 5$ be an integer. We claim that $G$ is a $C_\ell$ with vertices $v_0, v_1, \dots, v_{\ell-1}$, in cyclic order if and only if the set of triples of $G$ is equal to ${X = \{v_iv_{i+1}v_{i+2} \ |}$ $ \ i= 1,2,\dots, \ell \}$ where indices are taken modulo $\ell$.

The \say{if} statement is trivial. 
Suppose $G$ is a graph on $\ell$ vertices $v_0, v_1, \dots, v_{\ell-1}$ and the $X$ is the set of triples of $G$.
Suppose $v_1v_2$ is not an edge of $G$. Then $v_0, v_3$ must be adjacent to both $v_1$ and $v_2$ since $v_0v_1v_2$ and $v_1v_2v_3$ are both triples of $G$. But then $v_0 \dd v_2 \dd v_3$ is a subgraph of $G$, a contradiction since $v_0v_2v_3 \not \in X$.
\qed \end{proof}

\begin{observation}\label{lem:apex}
A graph $G$ is an $\ell$-wheel if and only if there is some vertex $v$ that appears in $\binom{\ell}{2}$ triples and $G \setminus v$ is a $C_{\ell}$.
\end{observation}

 By combining \cref{obs:P5C5} and \cref{lem:apex}, we obtain for each $\ell \geq 5$, we can reconstruct the $\ell$-wheel when we know the graph in question is triangulated, planar and $4$-connected. 
It follows from, \cref{obs:P5C5}, \cref{lem:apex} and \cref{lem:nbd-of-degree-6} that: 
\begin{corollary}
\label{corr:degree-five-vertex}
    If $G$ is known to be a triangulated planar graph without any separating triangles, for each $v \in G$ of degree at least five we can reconstruct the graph induced by $v$ and its neighbours.
\end{corollary}

\begin{lemma}\label{lem:reconstruct-4-conn-case-planar}
    Every connected graph $G$ on at least seven vertices that is known to be planar, triangulated and $4$-connected can be reconstructed. 
\end{lemma}
\begin{proof}
    It is an easy consequence of Euler's formula (given as \cref{corr:degree-five-vertex}), that $G$ must have a vertex $v$ of degree at least five.
    By \cref{corr:degree-five-vertex}, we can reconstruct the graph induced by $N[v]$ since $G$ is $4$-connected.
    We say the distance between two vertices $x,y$ is the number $d(x,y)$ of edges on a shortest path between them.
    Let $B_{m}(v)$ denote the graph induced by the set of vertices of distance at most $m$ from $v$.
    Suppose for some $\ell \geq 2$, we can reconstruct $B_{\ell-1}(v)$ from the connected triples of $G$.
    \begin{claim}
    We can determine the set of vertices $w \in N(x)$ that are of distance $\ell$ from $v$ for every vertex $x$ satisfying $d(v,x) = \ell-1$.\label{stmt:nbrlist}
    \end{claim}
    Since $\ell \geq 2$, there is an edge between $x$ and some vertex $y$ satisfying $d(v, y) = \ell -2$.
    Then, $q$ is a neighbour of $x$ of distance $\ell$ from $v$ if and only if $q \not \in V(B_{\ell -1}(v))$ and $qxy$ is a connected triple.
    Thus, we reconstruct the list of neighbours of $x$ of distance $\ell$ from $v$. This proves the claim.
    \\
    \\
    For a vertex $b \in G$ and integer $m$ let $N^{m}_v(b)$ denote the set of neighbours of $b$ of distance at most $m$ from $v$.
    Then by the previous claim we can reconstruct the set $N^{\ell -1}_v(w)$ for each $w \in G$. Moreover, we can determine the set of vertices of distance $\ell$ from $v$. 
    Let $u$ be a vertex of distance $\ell$ from $v$.
    We complete the proof by showing we can reconstruct the set $N^\ell_v(u)$.
    Let $z \neq u$ be a vertex of distance $\ell$ from $v$.

    \begin{claim}
        Let $z$ be a vertex of distance $\ell$ from $v$.
        Suppose that $z$ and $u$ are adjacent. Then, if $z$ and $u$ have a common neighbour $x$ satisfying $d(v,x) = \ell -1$ it follows that $x$ must have degree at least five. Moreover, the edge $zu$ can be reconstructed.
    \end{claim}

    Since $\ell \geq 2$, by definition $x$ must have a neighbour $x'$ of distance $\ell -2$ from $v$. Then $z,u$ are both not adjacent to $x'$.
    Since $G$ is triangulated, it follows that $x', x, u$ do not all occur together on a cycle bounding a face and $x', x, z$ do not all occur together on a cycle bounding a face.
    Because $G$ is planar there is a single face $C$ of $G \setminus x$ that contains all neighbours of $x$.
    Let $P_u$ be the $x'u$-path of $C$ that does not contain $z$ and let $P_z$ be the $x'z$-path of $C$ that does not contain $u$.
    Then, since $G$ is triangulated $x$ must have a neighbour in the interior of both $P_u$ and $P_z$.
    Hence, $x$ has degree at least five.
    So by \cref{corr:degree-five-vertex}, $G[N[x]]$ and thus the edge $uz$ can be reconstructed. 
    This proves the claim.
    \\
    \\
    We complete the proof by showing that we can determine the edges from $u$ to vertices $z$ satisfying $d(v,z) = \ell$ and $N^{\ell-1}_v(z) \cap N^{\ell -1}_v(u) = \emptyset$.
    Let $x$ denote a neighbour of $u$ that is of distance $\ell-1$ from $v$.
    Then $z$ and $u$ are adjacent if and only if $xuz$ is a triple.
\qed \end{proof}

We can also recognise \say{small enough} separators of a graph. We only need this result for separators of size three on graphs with at least six vertices.

\begin{observation}
    \label{lem:recognising-separators}
    Let $G$ be a connected graph on $n$ vertices and let $S \subseteq V(G)$ have cardinality at most $n -3$.
    Then $S$ is a separator if and only if it is possible to partition $V(G) \setminus S$ into two non-empty parts $V_1, V_2$  such that any triple containing a vertex from $V_1$ and a vertex from $V_2$ must contain a vertex from $S$. 
\end{observation}

Using the fact that we can detect separators of size three, we will show that all triangulated planar graphs of at least seven vertices have a unique reconstruction by performing recursion along separating triangles.
We will need the following observation.

\begin{observation}\label{obs:defining-the-recursion}
Let $G$ be a triangulated plane graph and let $S$ be a separator of size three in $G$.
Then $G \setminus S$ has exactly two components $V_1, V_2$ and $G[S\cup V_1]$ and $G[S \cup V_2]$ are both triangulated planar graphs. 
\end{observation}

\cref{obs:defining-the-recursion} follows from the fact that if $G$ is a triangulated planar graph every component of $G \setminus S$ would need to have an edge to each vertex of $S$.
Since $S$ is the vertex set of a triangle by \cref{obs:4conn} and planar graphs do not have $K_{3,3}$ subgraphs it follows that there can be at most two components of $G \setminus S$.
We can use \cref{obs:defining-the-recursion} to perform recursion on $G[(V(H) \cup V_1]$ and $G[V(H) \cup V_2]$ until we are left with either a triangulated planar graph with no separating triangles with at least seven vertices (which we can handle by \cref{lem:reconstruct-4-conn-case-planar}) or a triangulated planar graph on at most six vertices.
Next, we will show that if we end up with a triangulated planar graph on at most six vertices from this procedure, we are able to reconstruct it.
We begin by describing the triangulated planar graphs on at most six vertices.

For $\ell \geq 3$, we call a graph obtained from an $\ell$-wheel by adding as many chords as possible to the $C_\ell$ while keeping the graph planar a \emph{triangulated $\ell$-wheel}.
Note when $\ell \leq 5$, this is uniquely defined, up to labelling (see \cref{fig:6-vertex-triangulated-cases}).
Moreover, any triangulated planar graph on $n$ vertices containing a vertex of degree $n-1$ must be a triangulated $(n-1)$-wheel.

\begin{figure}[t]
    \centering
    \includegraphics[width=.7\textwidth]{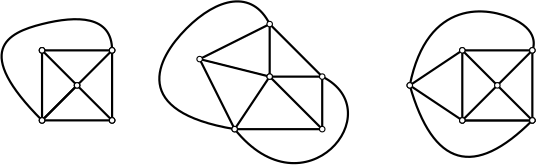}
    \caption{A depiction of all the distinct (unlabelled) triangulated planar graphs on five or six vertices. From left to right, the graphs are the triangulated 4-wheel, triangulated 5-wheel, and the octahedron.}
    \label{fig:6-vertex-triangulated-cases}
\end{figure}

\begin{observation}\label{obs:cases}
Every triangulated planar graph on at most six vertices is either a $K_3$, a $K_4$, a triangulated $4$-wheel, an octahedron or a triangulated $5$-wheel. 
\end{observation}

We are now ready to prove the base case of our recursion.

\begin{lemma}\label{lem:small-planar}
    Let $G$ be a graph on at least seven vertices that is known to be triangulated and planar.
    Suppose $G$ has a separator $S$ of size three and let $V_1$ be a component of $G \setminus S$ of size at most three.
    Then $G[S \cup V_1]$ can be reconstructed using the triples of $G$ and the knowledge of the labels of $S, V_1$ and the fact that $S$ is a separator of $G$.
\end{lemma}
\begin{proof}
    Let $G'$ denote the graph $G[S \cup V_1]$.
    Since $G \setminus V_1$ is non-empty and $G$ is a triangulated planar graph each vertex $s \in S$ must have a neighbour in $u_s \in G \setminus G'$.
    Then, we can determine the edges between $S$ and $V_1$ by checking if $su_sv$ is a triple for each $s \in S$ and $v \in V_1$.
    Hence, we need only show that we can reconstruct the edges between vertices of $V_1$.
    This is trivial when $G'$ has four vertices.
    When $G'$ has five vertices it must be 4-connected by \cref{obs:cases} so the two vertices in $V_1$ must be adjacent.
    
    When $G'$ has six vertices, it is either the octahedron or the triangulated 5-wheel.
    Every vertex in the octahedron has degree four and the triangulated 5-wheel has exactly two vertices of degree four.
    Thus, since we know the edges incident to $S$ in $G'$ we can determine whether $G'$ is an octahedron or a triangulated 5-wheel. Note that the graph obtained after removing the vertex set of any triangle in the octahedron is a triangle.
    Hence, $V_1$ must induce a triangle. So we can reconstruct $G'$ if it is an octahedron.

    We complete the proof by supposing $G'$ is a triangulated $5$-wheel and showing we can reconstruct $G'$. Note the triangulated $5$-wheel has exactly two vertices of degree 3, degree 4 and degree 5.
    Moreover, it has an isomorphism mapping the vertex of degree $i$ to the other vertex of degree $i$ for each $i \in \{3,4,5\}$.
    We say $S$ has degrees $(n_1,n_2,n_3)$ if the vertices of $S$ can be ordered so that the $i$-th vertex has degree $n_i$ in $G'$. 
    Because of the symmetry of $G'$, the unlabelled graph on $G' \setminus S$ is determined by the degrees of $S$.
    Since $S$ induces a triangle, there are only four possible degrees $S$ can have $(5, 5, 4)$, $(5,5,3)$, $(5,4,4)$, $(5,4,3)$.
    Moreover, since we already reconstructed the edges between $S$ and $G' \setminus S$, if there is some $x, y \in G' \setminus S$ such that $y$ has a neighbour $s \in S$ that is not adjacent to $x$ then we can test if $x$ and $y$ are adjacent by checking if $sxy$ is a triple.
    It is not difficult to see that these observations together are enough to reconstruct $G'$.
\qed \end{proof}

We are now ready to prove the main result of this subsection.
\repeattheorem{triangulated}
\begin{proof}
    Let $G$ be a triangulated planar graph on at least seven vertices.
    By \cref{lem:recognising-separators}, we can determine if $G$ has a separator of size $3$.
    If it doesn't, we can reconstruct $G$ by \cref{lem:reconstruct-4-conn-case-planar}. So, we may assume $G$ has a separator $S$ of size 3.
    Then $G \setminus S$ has exactly two components, and we can recognise their vertex sets $V_1, V_2$ by \cref{lem:recognising-separators} and \cref{obs:defining-the-recursion}.
    Moreover, by \cref{obs:defining-the-recursion} both $G[V_1 \cup S]$ and $G[V_2 \cup S]$ are planar and triangulated.
    For $i\in \{1,2\}$, if $G[V_i \cup S]$ contains at most six vertices we can reconstruct $G[V_i \cup S]$ by \cref{lem:small-planar}.
    Otherwise, we perform recursion on $G[V_i \cup S]$.
\qed \end{proof}

Note that, while our proof is constructive, it does not lead to a faster algorithm than the general cubic or \(O(n \cdot \norm{T})\) time algorithm described in \cref{sec:algorithm}.